\documentclass[review,number]{elsarticle}

\usepackage[margin=1.35in]{geometry}
\usepackage{amsmath,amsthm,amssymb}
\usepackage{color}

\newcommand{\BO}[1]{{O}\left(#1\right)}
\newcommand{\BT}[1]{{\Theta}\left(#1\right)}
\newcommand{\BOM}[1]{{\Omega}\left(#1\right)}

\newtheorem{lemma}{Lemma}
\newtheorem{theorem}{Theorem}

\renewenvironment{proof}{\begin{trivlist}
                         \item[] {\em Proof:}}{\hfill $\Box$
                       \end{trivlist}}

\newcommand{\X}{{X}}
\newcommand{\Y}{{Y}}
\newcommand{\Z}{{Z}}
\newcommand{\B}{\mathcal{B}}
\newcommand{\F}{{F}}
\newcommand{\SMerge}{$S$-\textsf{Merge}}
\newcommand{\SSort}{$S$-\textsf{Sort}}
\newcommand{\SPurifyingMerge}{$S$-\textsf{PurifyingMerge}}
\newcommand{\SBucketSort}{$S$-\textsf{BucketSort}}
\newcommand{\PurifyingMerge}{\textsf{PurifyingMerge}}
\newcommand{\NaiveSort}{\textsf{NaiveSort}}
\newcommand{\UnbalancedMerge}{\textsf{UnbalancedMerge}}
\newcommand{\Insert}{\textsf{Insert}}
\newcommand{\Deletemin}{\textsf{Deletemin}}
\newcommand{\Push}{\textsf{Push}}
\newcommand{\Pull}{\textsf{Pull}}
\newcommand{\Peekmin}{\textsf{Peekmin}}
\begin{document}
\begin{frontmatter}
\title{Exploiting non-constant safe memory in resilient algorithms and data structures}

\author{Lorenzo De Stefani}
\author{Francesco Silvestri\corref{cor1}}

\address[dei]{Dipartimento di Ingegneria dell'Informazione, University of Padova
\\ Via Gradenigo 6/B, I-35131 Padova, Italy\\
\{destefan,silvest1\}@dei.unipd.it}

\journal{Theoretical Computer Science}
\cortext[cor1]{Corresponding author. Phone number: +39 049 8277954}

\begin{abstract}
We extend the Faulty RAM model by Finocchi and Italiano (2008) by adding a safe
memory of arbitrary size $S$, and we then derive tradeoffs between the
performance of resilient algorithmic techniques and the size of the safe memory.
Let $\delta$ and $\alpha$ denote, respectively, the maximum amount of faults
which can happen during the execution of an algorithm and the actual number of
occurred faults, with $\alpha \leq \delta$. We propose a resilient algorithm for
sorting  $n$ entries which requires $\BO{n\log n+\alpha (\delta/S +  \log S)}$
time and uses $\BT{S}$ safe memory words. Our algorithm outperforms previous
resilient sorting algorithms which do not exploit the available safe memory and
require $\BO{n\log n+ \alpha\delta}$ time. Finally, we exploit our sorting
algorithm for deriving a resilient priority queue.  Our implementation uses
$\BT{S}$ safe memory words and $\BT{n}$ faulty memory words for storing $n$
keys, and requires $\BO{\log n + \delta/S}$ amortized time for each insert and
deletemin operation. Our resilient priority queue improves the $\BO{\log n +
\delta}$ amortized time required by the state of the art.
\end{abstract}


\end{frontmatter}

\section{Introduction}
Memories of modern computational platforms are not completely reliable since a
variety of causes, including cosmic radiations and alpha particles~\cite{B05},
may lead to a transient failure of a memory unit and to the loss or corruption
of its content. Memory errors are usually silent and hence an application may
successfully terminate even if the final output is irreversibly corrupted. This
fact has been recognized in many systems, like in Sun Microsystems servers at
major customer sites~\cite{B05} and in Google's server fleets~\cite{SPW11}.
Eventually, a few works have also shown that memory faults can cause serious
security vulnerabilities (see, e.g.,~\cite{GA03}).

As hardware solutions, like Error Correcting Codes (ECC), are costly and
reduce space and time performance, a number of algorithms and data structures
have been proposed that provide (almost) correct solutions even when silent
memory errors occur. Algorithmic approaches for dealing with unreliable
information have been widely targeted in literature under different settings,
and we refer to~\cite{FINO07} for a survey. In particular, a number of
algorithms and data structures, which are called \emph{resilient}, have been
designed in the \emph{Faulty RAM} (\emph{FRAM})~\cite{FI08}. In this model, an
adaptive adversary can corrupt up to $\delta$ memory cells of a large unreliable
memory at any time (even simultaneously) during the execution of an algorithm.
Resilient algorithmic techniques have been designed for many problems, including
sorting~\cite{FGI09}, selection~\cite{KT12}, dynamic programming~\cite{CFFS11},
dictionaries~\cite{FINO09}, priority queues~\cite{JMM07}, matrix multiplication
and FFT ~\cite{RE13}, K-d and suffix trees~\cite{GMV12,CDK11}. Resilient
algorithms have also been experimentally
evaluated~\cite{PI14,RE13,PetrilloGI13,PetrilloFI10}.

\subsection{Our results}
Previous results in the FRAM model assume the existence of a {safe memory}
of constant size which cannot be corrupted by the adversary and which is used for
storing crucial data such as  code and instruction counters. In this paper we
follow up the preliminary investigation in~\cite{CFFS11} studying to which
extent the size of the safe memory can affect the performance of resilient
algorithms and data structures. We enrich the FRAM model with a safe memory of
arbitrary size $S$ and then give evidence that an increased safe memory can be
exploited to notably improve the performance of resilient algorithms. In
addition to its theoretical interest, the adoption of such a model is supported
by recent research on hybrid systems that integrate algorithmic
resiliency with the (limited) amount of memory protected by hardware
ECC~\cite{LI13}. In this setting, $S$ would denote the memory that is protected
by the hardware.

Let $\delta$ and $\alpha$ denote respectively the maximum amount of faults which
can happen during the execution of an algorithm and the actual number of
occurred faults, with $\alpha \leq \delta$. In Section~\ref{sec:sort}, we show
that $n$ entries can be resiliently sorted in $\BO{n \log n + \alpha (\delta /S
+\log S)}$ time when a safe memory of size $\BT{S}$ is available in the FRAM. As
a consequence, our algorithm runs in optimal $\BT{n\log n}$ time as soon as
$\delta=\BO{\sqrt{n S\log n}}$ and $S\leq n/\log n$. When $S=\omega(1)$, our
algorithm outperforms previous resilient sorting algorithms, which do not
exploit non-constant safe memory and require $\BO{n\log n+ \alpha\delta}$
time~\cite{FGI09,KT12}. Finally, we use the proposed resilient sorting algorithm
for deriving a resilient priority queue in Section~\ref{sec:pqueue}.  Our
implementation uses $\BT{S}$ safe memory words and $\BT{n}$ faulty memory words
for storing $n$ keys, and requires $\BO{\log n + \delta/S}$ amortized time for
each insert and deletemin operation. This result improves the state of art for
which $\BO{\log n + \delta}$ amortized time is required for each
operation~\cite{JMM07}.

\subsection{Preliminaries}
As already mentioned, we use the FRAM model with a safe memory. Specifically,
the adopted model features two memories: the \emph{faulty memory} whose size is
potentially unbounded, and the \emph{safe memory} of size $S$. For the sake of
simplicity, we allow algorithms to exceed the amount of safe memory by a
multiplicative constant factor. The adversary can read the content of the faulty
memory and corrupt at any time memory words stored  in any position of the
faulty memory for up to a total $\delta$ times. Note that faults can occur
simultaneously and the adversary is allowed to corrupt a value which was already
previously altered. The safe memory can be read but not corrupted by the
adversary. A similar model was adopted in~\cite{CFFS11}, however in this paper
the adversary was not allowed to read the safe memory. We denote with
$\alpha\leq \delta$  the actual number of faults injected by the adversary
during the execution of the algorithm. Since the performance of our algorithms
do not increase as soon as $S>\delta$, we assume through the paper that $S\leq
\delta$; this assumption can be easily removed by replacing $S$ with $\min\{S,
\delta\}$ in our algorithms. 

A variable is \emph{reliably written} if it is replicated $2\delta+1$ times in
the faulty memory and its actual value is determined by majority: clearly, a
reliably written variable cannot be corrupted. We say that a value is
\emph{faithful} if it has never been corrupted and that a sequence is
\emph{faithfully ordered} if all the faithful values in it are correctly
ordered. Finally, we assume all faithful input values to be distinct, each
value to require a memory word, and that each sequence or buffer to be stored
in adjacent memory words.

\section{Resilient Sorting Algorithm}\label{sec:sort}
In the resilient sorting problem we are given a set of $n$ keys and the goal is
to correctly order all the faithful input keys (corrupted keys can be
arbitrarily positioned). We propose \emph{\SSort{}}, a resilient sorting
algorithm which runs in $\BO{n\log n+\alpha \left(\delta/S+\log S \right)}$ time
by exploiting $\BT{S}$ safe memory words. Our approach builds on the
resilient sorting algorithm in~\cite{FGI09}, however major changes are required
to fully exploit the safe memory. In particular, the proposed algorithm
forces the adversary to inject $\BT{S}$ faults in order to invalidate part of
the computation and to increase the running time by an additive
$\BO{\delta + S \log S}$ term. In contrast, $\BO{1}$ faults suffice to increase
by an additive $\BO{\delta}$ term the time of previous
algorithms~\cite{FI08,FGI09,KT12}, even when $\omega(1)$ safe memory is
available.
Our algorithm  runs in optimal $\BT{n \log n}$ time for
$\delta= \BO{\sqrt{S n\log n}}$ and $S\leq n/\log n$: this represents a 
$\BT{\sqrt{S}}$ improvement with respect to the state of the art
\cite{FGI09}, where optimality is reached for $\delta=\BO{\sqrt{n\log n}}$.

\SSort{} is based on mergesort and uses the resilient algorithm \emph{\SMerge{}}
for merging. The \SMerge{} algorithm requires $\BO{n+\alpha\left(\delta/S+\log
S\right)}$ time for merging two faithfully ordered sequences of length $n$ each
with $\BT{S}$ safe memory. \SMerge{} is structured as follows. An incomplete
merge of the two input sequences is initially computed with
\emph{\SPurifyingMerge{}}: this method returns a faithfully ordered sequence $Z$
of length at least $2(n-\alpha)$ that contains a partial merge of the input
sequences, and a sequence $\F$ with the at most $2\alpha$ remaining keys that
the algorithm has failed to insert into $Z$. Finally, keys in $\F$ are
inserted into $Z$ using the \emph{\SBucketSort{}} algorithm,
obtaining the final faithfully ordered sequence of all input values. 
Procedures \SPurifyingMerge{} and \SBucketSort{} are respectively
proposed in Sections~\ref{sec:purify} and~\ref{sec:bucket}, while
Section~\ref{sec:merge} describes the resilient algorithms \SMerge{} and
\SSort{}.

\subsection[SPurifyingMerge algorithm]{\SPurifyingMerge{}
algorithm}\label{sec:purify} 
Let $X$ and $Y$ be the faithfully ordered input sequences of length $n$ to be
merged. The \SPurifyingMerge{} algorithm returns a faithfully ordered sequence
$Z$ of length at least $2(n-\alpha)$ and a sequence $\F$ of length at most
$2\alpha$: sequence $Z$ contains part of the merging of $X$ and $Y$, while $F$
stores the input keys that the algorithm has deemed to be potentially corrupted
and has failed to insert into $Z$. The algorithm extends the \PurifyingMerge{}
algorithm presented in~\cite{FGI09} by adding a two-level cascade of
intermediate buffers, where the smallest ones are completely contained in the
safe memory. Specifically, the algorithm uses six support buffers\footnote{It
can be shown that a more optimized implementation of \SPurifyingMerge{} requires
only two buffers (i.e., $\X_2$ and $\Y_2$). However,  we describe here the
implementation with six support buffers for the sake of simplicity.}:
\begin{itemize}
 \item Buffers $\X_1$ and $\Y_1$ of length $4\delta+S$, and $\Z_1$ of length
$\delta+S/2$; they are stored in the faulty memory.
 \item Buffers $\X_2$ and $\Y_2$ of length $S$, and $\Z_2$ of length $S/2$;
they are stored in the safe memory.
\end{itemize}

\begin{figure}
\centering
\resizebox{.95\textwidth}{!}{
\input{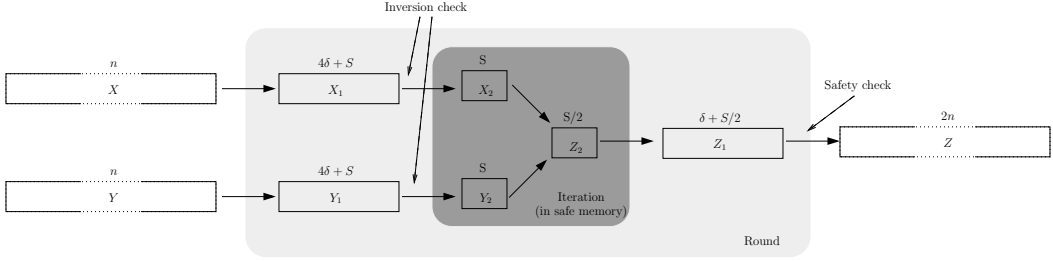}
}
\caption{Graphical representation of the \SPurifyingMerge{} algorithm. $X$ and
$Y$ are the input sequences to merge, and $Z$ is the output buffer. $X_1$, $Y_1$
and $Z_1$ are support buffers stored in the faulty memory, while $X_2$, $Y_2$
and $Z_2$ are support buffers stored in the safe memory.
The light gray (resp., dark gray) highlights the structures that are used in a
round (resp., iteration). An inversion (resp., safety) check is invoked any
time data are moved from $X_1$ to $X_2$ or from $Y_1$ to
$Y_2$ (resp., from $Z_1$ to $Z_2$). \label{fig:purify}}
\end{figure}

At high level, the algorithm works as follows (see Figure~\ref{fig:purify} for a
graphical representation). The computation is organized in \emph{rounds}. In
each round, $\BO{\delta}$ input keys in $X$ and $Y$ are respectively pumped into
buffers $\X_1$ and $\Y_1$. Then, the algorithm merges these keys in $\Z_1$ by
iteratively merging small amounts of data in safe memory: during each
\emph{iteration}, chunks of $\BO{S}$ consecutive keys in $\X_1$ and $\Y_1$ are
moved into buffers $\X_2$ and $\Y_2$, where they are merged in $\Z_2$ using a
standard merging algorithm.  Keys in $\Z_2$ are shifted into $\Z_1$ at the  end
of each iteration, while keys in $\Z_1$ are appended to $Z$ at the end of each
round. The algorithm performs some checks, which are explained in details later,
when keys are moved among buffers in order to guarantee resiliency: an
\emph{inversion check} is done every time a key is shifted from $\X_1$ to $\X_2$
or from $\Y_1$ to $\Y_2$;  a \emph{safety check} is executed every time buffer
$\Z_1$ is appended to $\Z$. If a check is unsuccessful, some critical faults
have occurred and then part of the computation must be rolled back and
re-executed. 

We now provide a more detailed description. Each round starts by filling buffers
$\X_1$ and $\Y_1$ with the remaining keys in $X$ and $Y$, starting from those
occupying the smallest index positions (i.e., from the smallest faithful
values). Subsequently, the algorithm fills $\Z_1$ with at least $\delta$ values
from the sequence obtained by merging $\X_1$ and $\Y_1$ or until there are no
further keys to merge. Specifically, buffer $\Z_1$ is
filled by iterating the following steps until it contains at least $\delta$
values  or there are no further keys to merge in $\X_1$, $\X_2$, $\Y_1$ and
$\Y_2$:
\begin{enumerate}
\item Buffers $\X_2$ and $\Y_2$ are filled with the remaining keys of
$\X_1$ and $\Y_1$, respectively, starting from the smallest index position. With
the exception of the first iteration of the first round, an  inversion check is
executed for each key inserted in $\X_2$ and $\Y_2$. If a
check is unsuccessful, the current round is restarted. In the first iteration of
the first round, each key is inserted in $\X_2$ and $\Y_2$  without any check.

\item Buffers $\X_2$ and $\Y_2$ are merged in $\Z_2$, until buffer $\Z_2$ is
full or there are no further entires in $\X_2$ and $\Y_2$. 
The merging is performed using the standard algorithm since input and
output buffers are stored in safe memory.

\item Buffer $\Z_2$ is appended to $\Z_1$ and then emptied.
\end{enumerate}

As soon as $\Z_1$ is full or there are no further keys, a safety check is
performed on $\Z_1$: if it succeeds,  buffer $\Z_1$ is appended to $Z$ and
flushed and then a new round is started; otherwise, the current round is
restarted. 

The inversion check works as follows. The check is performed on every new key
$x$ of $\X_1$ inserted into $\X_2$, and on every new key $y$ of $\Y_1$ inserted
into $\Y_2$. 
We describe the check performed on each entry $x$, being the control executed
on $y$ defined correspondingly. 
If $\X_2$ is empty, no operation is done and the check ends successfully.
Otherwise, the value $x$ is compared with the last inserted key $x'$ in $\X_2$.
If $x$ is larger than $x'$, no further operations are done and the check ends
successfully. Otherwise, if $x$ is smaller than or equal to $x'$, it is possible
to conclude that at least one of the two keys is corrupted since $\X_2$ is
supposed to be faithfully ordered and each key to be unique. Then, both keys are
inserted into $\F$ and removed from $\X_1$ and $\X_2$; if there exists at least
one value in $\X_2$ after the removal, the check ends successfully, and it
ends unsuccessfully otherwise. We observe that inversion
checks guarantee $\X_2$ and $\Y_2$ to be perfectly ordered at any time (recall
that the two buffers are stored in safe memory).

The safety check works as follows. The check is performed when $\Z_1$
contains at least $\delta$ keys or there are no more keys to merge. In the last
case, the check always ends successfully. Suppose now that $\Z_1$ contains at
least $\delta$ keys, and let $z$ be the latest key inserted into $\Z_1$ which we
assume to be stored in safe memory. Denote with $\X'$ (resp., $\Y'$) the
concatenation of keys in $\X_2$ and $\X_1$ (resp., $\Y_2$ and $\Y_1$). If there
are less than $S/2$ keys in $\X'$ and $\Y'$  smaller than or equal to $z$, the
safety check ends successfully. Otherwise, the algorithm scans $\X'$ starting
from the smaller position and compares each pair of adjacent keys looking for
inversions: if a pair is not ordered, it is possible to conclude that at least
one of the two values has been corrupted, and hence both keys are  inserted in
$\F$ and removed from $\X'$.  A similar procedure is executed for $\Y'$ as well.
The check then ends unsuccessfully. 

When a round is restarted due to an unsuccessful check, 
the algorithm replaces keys in $\Z_i$, $\X_i$ and $\Y_i$, for any
$i\in \{1,2\}$, with the keys contained in the respective buffers at beginning
of the round (specifically, just after the algorithm terminates to fill
buffers $\X_1$ and $\Y_1$  with
new keys). However, keys
that have been moved to $F$ during the failed round are not restored
(empty positions in $\X_1$ and $\Y_1$ are suitably filled with keys in $X$
and $Y$). This operation can be implemented by storing a copy of $\X_1$, $\Y_1$
and $\Z_1$ in the faulty memory, and of $\X_2$, $\Y_2$ and $\Z_2$ in the safe
memory. For every key moved to $F$, the key is also removed from the
copies\footnote{An entry can be removed from a sequence in constant time by
moving the subsequent entries in the correct position as soon as they are read
and by maintaining the required pointers in the safe memory.}.

\begin{lemma}\label{lem:purify}
Let $X$ and $Y$ be two faithfully ordered sequences of length $n$.
\SPurifyingMerge{} returns a faithfully ordered sequence $Z$ of length $\vert Z
\vert \geq n- 2\alpha$ containing part of the merge of $X$ and $Y$, and a
sequence $F$ of length $\vert F \vert \leq 2\alpha$ containing the remaining
input keys. The algorithm runs in $\BO{n+\alpha\delta/S}$ time and uses
$\BT{S}$ safe memory words.
\end{lemma}
\begin{proof}
It is easy to see that the algorithm uses $\BT{S}$ safe memory and that each
input key must be in $Z$ or $F$. We prove that $Z$ is faithfully ordered as
follows: we first show that $\Z_1$ is faithfully ordered at the end of each
round; we then argue that $\Z_1$ can be appended to $Z$ without affecting the
faithful order of $Z$. We say that a round is successful if the round is not
restarted by an unsuccessful  inversion or safety check. For proving the
correctness of the algorithm we focus on successful rounds since unsuccessful
ones do not affect $Z$. 

Let us now show that buffer $\Z_1$ contains a faithfully ordered sequence at the
end of a successful round. Inversion checks guarantee that any key inserted in
$\X_2$ is not smaller than the previous one, and then buffer $\X_2$ is sorted at
any time. Moreover, since the round is successful, buffer $\X_2$ always contains
at least one key during the round and, in particular, the buffer contains a key
between two consecutive iterations. This fact guarantees that the concatenation
$\hat X$ of \emph{all} keys inserted in $\X_2$ in each iteration of the round
creates an ordered sequence. Similarly, we have that the concatenation $\hat Y$
of \emph{all} keys inserted in $\Y_2$ in each iteration of the round is ordered.
Each iteration of the round merges in safe memory a part of $\hat X$ and $\hat
Y$. Then. the concatenation of all keys written into $\Z_2$ during the round is
the correct merge of $\hat X$ and $\hat Y$ (note that the largest keys in $\hat
X$ and $\hat Y$  are not merged and are kept in $\X_2$ and $\Y_2$).  Since these
output keys are first stored in $\Z_2$ and then in $\Z_1$, we can claim that
$\Z_1$ is a faithfully ordered sequence: indeed, there can be an out-of-order
key in $\Z_1$ due to a corruption occurred after the key has been moved from
$Z_2$ to $\Z_1$.

We now prove that $Z$ is faithfully ordered. If the algorithm ends in one
successful round, then $Z$ is faithful ordered by the previous claim. We now
suppose that there are at least two successful rounds. Let $\Z^i_1$ be the two
buffers appended to $Z$ at the end of the $i$-th and $(i+1)$-st successful
round. We have that $\Z^i_1$ and $\Z^{i+1}_1$ are faithfully ordered by the
previous claim, and we now argue that even the concatenation of $\Z_1^i$ and
$\Z_1^{i+1}$ is faithfully ordered. Let $z$ be  the latest value $z$ inserted in
$\Z^i_1$: since $z$ is maintained in safe memory, we have that $z$ is larger
than all the faithful values in $\Z^i_1$ and smaller than all values in $\X_2$
and $\Y_2$. Since the round is successful, there are at most $S/2$ keys smaller
than or equal to $z$ in $\X'$ and hence at least $\delta+1$ keys larger than $z$
in $\X'$: indeed, the $4\delta+S$ entries in $\X_1$ at the beginning of the
round can be moved in $F$ (at most $2\delta$), in $\Z_1$ (at most
$\delta+S/2-1$), or in $\X_2$ (thus remaining in $\X'$).  Therefore, there
exists a faithful key in $\X'$ larger than $z$, and hence all the faithful
values remaining in $\X$ must be larger than $z$. The at most $S/2$ keys smaller
than or equal to $z$ must be in $\X_1$, and they will be removed by inversion
checks in subsequent rounds since there are at least $S/2$ values in $\X_2$
larger than $z$ (note that there are always at least $S/2$ keys $\X_2$ at the
end of an iteration). Since a similar claim applies to $Y$, we have that all
faithful keys in $\Z^{i+1}_1$ are larger than those in $\Z^i_1$. It follows that
$Z$, which is  the concatenation of $\Z^1_1, \Z^2_1,\ldots$, is faithfully
ordered.

Finally, we upper bound the running time of the algorithm.
If no corruption occurs, the algorithm requires $\BO{n}$ time since there are
$\BO{n/\delta}$ rounds, each one requiring $\BO{\delta/S}$ iterations of cost
$\BO{S}$. Faults injected by the adversary during the first iteration of
the first round cannot restart the round, but increase the running time by at
most a factor $\BO{\alpha}$ since each fault can cause two keys to be moved in
buffer $F$. Consider an unsuccessful round that fails due to an inversion check.
Since at the beginning of each iteration there are at least $S/2$ keys in both
$\X_2$ and $\Y_2$, then at least $S/2$ inverted pairs are required for emptying 
$\X_2$ or $\Y_2$ and the adversary must pay $S/2$ faults for getting an
unsuccessful inversion check. Consider now an unsuccessful round that fails due
to a safety check.  Since there are at leas $S/2$ keys larger than or equal to
$z$ in $\X_2$ and $\Y_2$ and there are at least $S/2$ keys smaller than $z$ in
$\X_1$ and $\Y_1$, there must be at least $S/2$ inversions. Since at least
$S/2$ of these inversions are removed during the safely check and cannot
be used by the adversary for failing another round,  the adversary
must pay  $S/2$ faults for getting an
unsuccessful safety check. In all cases, each unsuccessful round costs
$\BO{\delta}$ time to the algorithm and $S/2$ faults to the adversary: since
there cannot be more than $\lfloor 2\alpha /S\rfloor$ unsuccessful rounds, the
overhead due to unsuccessful rounds is $\BO{\alpha\delta/S}$. The lemma follows.
\end{proof}

\subsection[SBucketSort Algorithm]{\SBucketSort{} Algorithm}\label{sec:bucket}
Let $X$ be a faithfully ordered sequence of length $n_1$ and $Y$ an arbitrary
sequence of length $n_2$. The \SBucketSort{} algorithm computes a faithfully
ordered sequence containing all keys in $X$ and $Y$ in
$\BO{n_1+(n_2+\delta)n_2/S + (n_2+\alpha) \log S}$ time using a safe memory of
size $\BT{S}$. This algorithm extends and fuses the \NaiveSort{} and
\UnbalancedMerge{} algorithms presented in~\cite{FGI09}. 

The algorithm consists of $\lceil n_2/S \rceil$ rounds. At the beginning of each
round, the algorithm removes the $S$ smallest keys among those remaining in $Y$
and stores them into an ordered sequence $P$ maintained in safe memory.
Subsequently, the algorithm scans the remaining keys of $X$, starting from the
smallest position, and partitions them among $S+1$ buckets $\B_i$ where $\B_0$
contains keys in $(-\infty, P[1])$, $\B_i$ keys in $[P[i],P[i+1])$ for $1\leq i<
S$ and $\B_S$ keys in $[P[S],+\infty)$. The scan of $X$ ends  once $\delta+1$
keys have been inserted into $\B_S$ or there are no more keys left in $X$. For
each key $x$ in $X$, the search of the bucket is crucial for improving
performance and proceeds as follows: the algorithm checks if $x$ belongs to the
range of the last used bucket $\B_k$, for some $0\leq k\leq S$; if the check
fails, then the algorithm verifies if $x$ belongs to the right/left $\log S$
adjacent buckets of $\B_k$; if this search is again unsuccessful, the correct
buffer is identified by performing a binary search on $P$. When the correct
bucket is found, entry $x$ is removed from $X$ and appended to the sequence of
keys already in the bucket, thus guaranteeing that each bucket is faithfully
ordered. When the scan of $X$ ends, the sequence given by the concatenation of
$\B_0, P[1], \B_1, P[2],\ldots, \B_{S-1}, P[S]$ is appended to the output
sequence $Z$, keys in $\B_S$ are inserted again in $X$ (in suitable empty
positions of $X$ that maintain the faithful order of $X$), $P$ and the $S+1$
buckets $\B_i$ are emptied and a new round is started. After the  $\lceil n_2/S
\rceil$ rounds, all remaining keys in $X$ are appended to $Z$.

The $S$ smallest values of $Y$, which are used in each round for determining the
buckets, are extracted from $Y$ using the following approach. At the beginning
of \SBucketSort{} (i.e., before the first round), a support priority queue
containing $S$ nodes is constructed in safe memory as follows. Keys in $Y$ are
partitioned into $S$ segments of size $\lceil n_2 /S \rceil$ and a node
containing the smallest key and a pointer to the segment is inserted
into the priority queue (the segment is stored in the faulty memory). Each time
the smallest value is required, the smallest value $y$ of the queue is
extracted (we note that $y$ is the minimum faithful key among those in $Y$ or a
corrupted key even smaller). Then, $y$ is removed from the queue and from the
respective buffer, and a new node with the new smallest key and a pointer to the
segment is inserted in the queue. When a value is removed from a segment, the
remaining values in the segment are shifted to keep keys stored in
consecutive positions. At the beginning of each round, the $S$ values
used for defining the buckets are obtained by $S$ subsequent extractions of the
minimum value in the priority queue and maintained in $S$ safe memory words.

\begin{lemma}\label{lem:bucket}
Let $X$ be a faithfully ordered sequence of length $n_1$ and let $Y$ be a
sequence of length $n_2$. \SBucketSort{} returns a faithfully ordered sequence
containing the merge of keys in $X$ and $Y$. The algorithm runs in
$\BO{n_1+(n_2+\delta)n_2/S + (n_2+\alpha) \log S}$ time and uses  $\BT{S}$ safe
memory words.
\end{lemma}
\begin{proof}
It is easy to see that the algorithm uses $\BT{S}$ safe memory and that each key
of $X$ and $Y$  must be in the output sequence at the end of the algorithm. We
now argue that the output sequence $Z$ is faithfully ordered at the end of the
algorithm: we first prove that the sequence appended to $Z$ at the end of each
round is faithfully ordered, and then  show that it can be appended to $Z$
without affecting the faithful order of $Z$. 
We now prove that the sequence
appended to $Z$ at the end of a round  is faithfully ordered. For each $0\leq i
\leq S$, each faithful key appended to bucket $\B_i$ is in the correct range
(note that $P$ cannot be corrupted being in safe memory) and the sequence of
keys in $\B_i$ is faithfully ordered since the insertion maintains the order  of
faithful keys in $X$. Therefore, the sequence $\B_0, P[1], \B_1, P[2],\ldots,
\B_{S-1}, P[S]$ that is appended to $Z$ is faithfully ordered. We now show that
the appended sequence guarantees the faithful order of $Z$ by proving that, at
the end of a round, the  faithful keys that remain in $X$ are larger than
those already in $Z$ (i.e., faithful keys that are appended in subsequent rounds
are larger). If the round ends since there are no further keys in $X$, the
claim is trivially true. Suppose now that the round ends since there are
$\delta+1$ keys in $\B_S$. Then, there must exist at least one faithful key
larger than $P[S]$ in $\B_S$ and all remaining faithful values in $X$ must be
larger than $P[S]$. Therefore, we can conclude that $Z$ is faithfully
ordered at the end of the algorithm: by the above argument $Z$ is faithfully
ordered at the end of the last round; the keys in $X$ that are appended to $Z$
after the last round do not affect the faithful order of $Z$ since  $X$ is
faithfully ordered and faithful keys in $X$ are larger than those already in
$Z$.

We now upper bound the running time. Suppose no corruptions occur during the
execution of the algorithm. Extracting the $S$ smallest keys from $Y$ using the
auxiliary priority queue requires $\BO{n_2+S\log S}$ time ($\BO{n_2 \log n_2}$
time if $n_2<S$) for each of the $\lceil n_2/S\rceil$ rounds, and therefore a
total $\BO{n_2 (n_2/S +\log S)}$ time. The insertion of an entry $X[i]$ in a
bucket, for
each $1\leq i\leq n_1$, requires $\BO{1+\min\{f_i, \log S\}}=\BO{1+f_i}$, where
$f_i$ is the number of keys of $Y$ in the range $(X[i-1], X[i])$: indeed, the
algorithm searches the bucket for $X[i]$ among the $\BO{\min\{f_i, \log S\}}$
buckets around
the
one containing $X[i-1]$ and then, only in case of failure, performs a binary
search on $P$. When no fault occurs, each key of $Y$ contributes to one of the
$f_i$'s  since $X$ is sorted. Therefore,  the partitioning of $X$ costs
$\BO{\sum_{i=1}^{n_1} (1+f_i)} = \BO{n_1 + n_2}$. We note that the above
analysis ignores the fact that in each round $\delta+1$ values are moved back
from $\B_{S}$ to $X$ this fact
leads to an overall increase of the running time given by an additive component
$\BO{\delta \lceil n_2/S\rceil}$, which follows by charging $\BO{\delta}$
additional operations to each round. A fault in $X$ may affect the running time
required for partitioning $X$. In particular, each fault may force the algorithm
to pay $\BO{\log S}$ for the corrupted key and the subsequent one in $X$:
indeed, a corruption of $X[i]$ may force the algorithm to perform a binary
search in order to find the right bucket for $X[i]$ and for the subsequent
key $X[i+1]$. The additive cost due to $\alpha$ faults is hence $\BO{\alpha \log
S}$. The corruption of keys in $Y$ does not affect the running time since the
algorithm does not exploit the ordering of $Y$. The lemma follows.
\end{proof}

\subsection[SMerge and SSort Algorithms]{\SMerge{} and \SSort{}
Algorithms}\label{sec:merge} As previously described, \SMerge{} processes the
two input sequences with \SPurifyingMerge{} and then the two output sequences
are merged with \SBucketSort{}. We get the following lemma.
\begin{lemma}\label{lem:merge}
Let $X$ and $Y$ be two faithfully ordered sequences of length $n$. Algorithm
\SMerge{} faithfully merges the two sequences in $\BO{n + \alpha \left(\delta/S
+ \log S\right)}$ time  using $\BT{S}$ safe
memory words. 
\end{lemma}
\begin{proof}
By Lemma~\ref{lem:purify}, algorithm \SPurifyingMerge\ returns a faithful
sequence $Z$ of length at most $2n$ and a
sequence $F$ of length at most $2\alpha$ in $\BO{n+\alpha \delta /S}$ time.
These output sequences are then combined using the \SBucketSort{} algorithm:
by Lemma~\ref{lem:bucket}, this algorithm  returns a faithfully ordered
sequence of all the input elements in $\BO{n+\alpha \left(\delta /S +\log
S\right)}$ time. The lemma follows.
\end{proof}

By using \SMerge{} in the classical mergesort algorithm\footnote{The standard
recursive mergesort algorithm requires a stack of length $\BO{\log n}$ which
cannot be corrupted. However, it is easy to derive an iterative algorithm where
a $\BT{1}$ stack length suffices.}, we get the desired resilient sorting
algorithm \SSort{} and the following theorem.
\begin{theorem}
Let $X$ be a sequence of length $n$. Algorithm \SSort{} faithfully sorts the
keys in $X$ in $\BO{n \log n+  \alpha  \left(\delta/S + \log S\right)}$ time
using $\BT{S}$ safe memory words.
\end{theorem}
\begin{proof}
Let us assume, for the sake of simplicity, $n$ to be a power of two, and denote
with $\alpha_{i,j}$ the number of faults that are detected by \SMerge\ on the
$j$-th recursive problem which operates on input sequences of length $2^i$, with
$0\leq i < \log n$ and $0\leq j < n/2^i$.
A fault injected in one sub-problem at level $i$ may affect the parent
problem at level $i+1$, but cannot affect sub-problems at level $i+2$. Indeed,
a key $x$ corrupted during the sub-problem at level $i$ may be out-of-order in
the output sequence. Key $x$ is then recognized by the \SMerge{} at
level $i+1$ as a fault, inserted in $F$ by \SPurifyingMerge{}, and then
positioned in the correct order in the output sequence by
\SBucketSort{} ($x$ will thus be consider as a faithful key in the parent
problem at level $i+2$). Another fault might cause key $x$ to be stored
out-of-order
again in the output sequence at level $i+1$, but this fact is accounted to the
new fault. Hence, we get $\sum_{i=0}^{\log n-1} \sum_{j=0}^{2^i-1} \alpha_i
\leq 2\alpha$. By the upper bound on the time of \SMerge\
in Lemma~\ref{lem:merge}, we get that the running time of \SSort\ is
upper bounded by 
$$ \BO{\sum_{i=0}^{\log n-1} \sum_{j=0}^{2^i-1}
\left(n/2^i+\alpha_{i,j} \left(\delta /S + \log S \right)\right)}.$$ 
The correctness of \SSort\ follows by the correctness of \SMerge.
\end{proof}

\section{Resilient Priority Queue}\label{sec:pqueue}
A \emph{resilient priority queue} is a data structure which 
maintains a set of keys that can be managed and accessed through two main 
operations: \Insert{}, which allows to add a key to the queue; and \Deletemin{}, 
which returns the minimum faithful key among those in the priority queue or
an even smaller corrupted key and then removes it from the priority queue. 

In this section we present an implementation of the resilient priority queue
that exploits a safe memory of size $\BT{S}$. Let $n$ denote the number of keys
in the queue. Our implementation requires $\BO{\log n + \delta / S}$ amortized
time per operation, $\BT{S}$ words in the safe memory and $\BT{n}$ words in the
faulty memory. 
Our resilient priority queue  is based
on the fault tolerant priority queue proposed in~\cite{JMM07}, which is in turn
inspired by the cache-oblivious priority queue in~\cite{ArgeBDHM07}. The
performance of the resilient priority queue is here improved by exploiting the
safe memory and the \SMerge{} and \SSort{} algorithms, in place of the resilient
merging and sorting algorithms in~\cite{FGI09}. It is important to point out
that the $\BOM{\log n + \delta}$ lower bound in~\cite{JMM07} on the performance
of the resilient priority queue does not apply to our data structure since the
argument assumes that keys are not stored in safe memory between operations. 
The amortized time of each operation in our implementation matches the
performance of classical
optimal priority queues in the RAM model when the number of tolerated
corruptions  is $\delta=\BO{S \log n}$: this represents a 
$\BT{S}$ improvement with respect to the state of the art
\cite{JMM07}, where optimality is reached for $\delta=\BO{\log n}$.

The presentation is organized as follows: we first present in
Section~\ref{sec:struct} the details of the priority queue implementation, with
particular emphasis on the role played by the safe memory; then we proceed in
Section~\ref{sec:corr} to prove its correctness and complexity bounds.

\subsection{Structure}\label{sec:struct}
The structure of our resilient priority queue is similar to the one used
in~\cite{JMM07}, however we require some auxiliary structures and different
constraints in order to exploit the safe
memory. Specifically, the resilient priority queue presented in this paper
contains the following structures (see Figure~\ref{fig:pqueue} for a graphical
representation):

\begin{figure}
  \centering
\resizebox{.95\textwidth}{!}{
\input{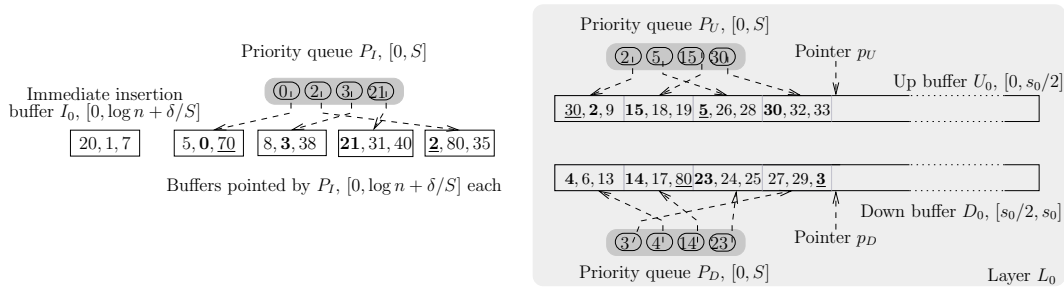}
}
  \caption{\label{fig:pqueue}Main support structures of the resilient priority
queue (layers $L_1, \ldots, L_{k-1}$ are omitted). The figure shows the
immediate insertion buffer $I_0$, the priority queue $P_I$ and the respective
pointed buffers, and layer $L_0$. Layer $L_0$ consists of the up buffer $U_0$,
the down buffer $D_0$, the priority queues $P_U$ and $P_D$, and the pointers
$p_U$ and $p_D$. The first $\delta+1$ entries of $U_0$ and of $D_0$ are
organized in up to $S$ sub-buffers of maximum size $\delta/S+1$. Each
sub-buffer of $U_0$ (resp., $D_0$) is pointed by a node in $P_U$ (resp., $P_D$).
All priority queues are stored in the safe memory, while buffers are
contained in the faulty memory. 
Underlined keys are corrupted, while bold keys are used as priority in some
queue.
Each range $[x,y]$ gives the minimum and maximum
number of contained keys in a buffer or nodes in a priority queue.
}
\end{figure}

\begin{itemize}
 \item The \emph{immediate insertion buffer} $I_{0}$,
which contains up to $\log n + \delta / S$ keys. This buffer is stored in the
faulty memory.
 
 \item The \emph{priority queue $P_{I}$}, which contains up to $S$ nodes. Each
node contains a pointer to a buffer of size at most $\log n + \delta/S$ and
the priority key of
the node is the smallest value in the pointed buffer. Buffers are
stored in the faulty memory, while the actual priority queue $P_I$ and other
structural information (e.g., buffer length, the position in the buffer of the
smallest key) are stored in $\BO{S}$ safe
memory
 words. The purpose of $P_{I}$ is to act as a buffer between the newly inserted
data in $I_0$ and the main structure of the priority queue, that is layers
$L_{0},\ldots,L_{k-1}$ (see below). On the one hand it allows to rapidly  access
the newly inserted keys, while on the other hand it accumulates such keys so
that the computational cost necessary for inserting all these keys in the main
structure is amortized over the insertion of at least $S\log n + \delta$ new
values.  
 
 \item The \emph{layers $L_{0},\ldots,L_{k-1}$}, with  $k = \BO{\log n}$. Each
layer $L_i$ contains two faithfully ordered buffers $U_{i}$  and $D_{i}$,
 named \emph{up buffer} and \emph{down buffer}, respectively. Up and
down buffers
are connected by a doubly linked list: for each $0\leq i <k$, buffer $U_i$ is
linked to $D_{i-1}$ and $D_i$ and vice versa. The layers are stored in the
faulty memory, while the size and the links to the neighbors of each buffer are
reliably written (i.e., replicated $2\delta+1$ times) in the faulty memory
using additional $\BT{\delta}$ space. For each layer, we define a threshold
value $s_{i}=2^{i+1}\left(S \log^{2} n + \delta \left( \log S + \delta/ S
\right)
\right)$ which is used to determine whether an up buffer $U_{i}$ has too many
keys or a down buffer $D_{i}$ has too few. Specifically, we impose the
following order and size invariants on all up and down buffers at any time:
\begin{itemize}
 \item (I1)  All buffers are faithfully ordered; 
 \item (I2)  For each $0 \leq i < k-1$, the concatenations $D_{i}D_{i+1}$ and
$D_{i}U_{i+1}$ are faithfully ordered;
 \item (I3) For each $0 \leq i < k-1$, $s_{i}/2 \leq \vert D_{i} \vert \leq
s_{i}$ (this invariant may not hold for the last layer); 
\item (I4) For each $0 \leq i < k$, $\vert U_{i} \vert \leq s_{i}/2$.
\end{itemize}

\item The \emph{priority queues $P_{U}$} and \emph{$P_{D}$} and the pointers
\emph{$p_U$} and \emph{$p_D$}, which are stored in the safe memory. These queues
are used to speed up the access to
entries in $U_0$ and $D_0$. We consider the buffer $U_{0}$ as the concatenation
of two buffers $U^{P}_{0}$ and $U^{S}_{0}$. $U^{S}_{0}$ contains  keys in the
$\delta +1$ smallest positions (if any) of $U_0$, while $U^{P}_{0}$ contains all
the remaining keys (if any) in $U_0$. $U^{S}_{0}$ itself is divided into up to
$S$
sub-buffers, each one with maximum size ${\delta/S+1}$ and associated with one
node of $P_U$: each node maintains a pointer to the beginning of a sub-buffer of
$U^{S}_{0}$ and its priority key is the smallest value in the respective
sub-buffer. Each node also contains support information such as the size of the relative sub-buffer
and the position of the smallest key in the sub-buffer.
$p_U$ points to the first element of $U^{P}_{0}$. This
structure ensures that the concatenation of a resiliently sorted $U^{S}_{0}$
with $U^{P}_{0}$ is a faithful ordering of all the elements in $U_0$ at all
times.  The priority queue $P_U$ can be built by determining the minimum element
of each sub-buffer in $U^{S}_{0}$ and then by building the
priority queue in safe memory. The priority queue $P_{D}$ and pointer $p_D$ are
analogously constructed from buffer $D_0$.
\end{itemize}

Since the priority queues $P_{I}$, $P_{U}$ and $P_{D}$ are resiliently stored,
we use any standard implementation that supports the \Peekmin{} operation, which
is an operation that returns the minimum value in the priority queue without
removing it. We note that buffer sizes (i.e., $s_i$) depend on $n$: As
suggested in~\cite{JMM07}, a global rebuilding of the resilient priority queue
is performed when the number of keys in it varies by $\BT{n}$. The rebuilding is
done by resiliently sorting all the keys and then distributing them among the
down buffers starting from $D_0$. 

The functioning and purpose of the auxiliary structures will
be detailed in the description  of the \Insert{} and \Deletemin{} operations in
Section~\ref{sec:insdel}.
We now provide an intuitive explanation of the
functioning of our priority queue. Newly inserted keys are collected in the
immediate insertion buffer $I_0$ and in the buffers pointed by nodes in $P_I$,
while 
the majority of the previously inserted values are maintained in the
up and down buffers in the $k$ layers $L_i$. 
The role of the down buffers is to contain small keys that are likely to be soon
removed by \Deletemin{} and then should move towards the lower levels (i.e.,
$I_0$, $P_I$ or $L_0$); on
the other hand, up buffers store large keys that will not be required in the
short time (note that this fact is a consequence of invariant (I2)).
Keys  are moved among
layers by means of the two fundamental primitives \Push{} and \Pull{}: these
functions, which are described in Section~\ref{sec:pushpull}, are invoked when
the up and down buffers violate the size invariants, and exploit the resilient
merging algorithm \SMerge{}. 
The purpose of the support structures is to reduce
the overhead necessary for the management of the priority queue in the presence
of errors by reducing the number of invocations to the costly
maintenance tasks (i.e., \Push{} or \Pull{}) and by amortizing their
computational cost over multiple executions of \Insert{} or \Deletemin{}.
It will be evident in the subsequent section that $P_U$ and $P_D$ may
cause a discrepancy with respect to the order invariants (I1) and (I2) for
the first $\BO{\delta}$ positions of buffers $D_{0}$ and $U_{0}$. However, we
will see that this violation can be in general ignored and can be quickly
restored any time the algorithm needs 
to exploit the invariants on $D_{0}$ and $U_{0}$, that is any time  \Push{} and
\Pull{} are invoked.

\subsubsection{\Insert{} and \Deletemin{}}\label{sec:insdel}
The implementation of \Insert{} and \Deletemin{} varies significantly with
respect to the resilient priority queue presented in~\cite{JMM07}. In
particular, the safe memory plays an important role  in order to
obtain the desired performance. 

\paragraph{\Insert{}} The newly inserted key is appended to the {immediate
insertion buffer} $I_{0}$. If after the insertion $I_{0}$ contains $\log n +
\delta / S$ keys, some values in $I_0$ are moved into other buffers as follows.
Suppose that $P_{I}$ contains less than $S$ nodes. A new buffer $I'$ is created
in the faulty memory and filled with the $\log n + \delta / S$ keys in
$I_{0}$, then a new node is inserted in $P_I$ with the minimum value in $I'$ as
key and a pointer to $I'$; $I_0$ is flushed at the end of this operation.
Suppose now that  $P_{I}$ contains $S$ nodes. All keys in buffer $I_0$, in the
buffers
pointed by all nodes of $P_I$ and in the sub-buffers managed through $P_U$
(i.e., in buffer $U^S_0$) are resiliently sorted using the \SSort{} algorithm.
These values are then merged with those in $U^P_{0}$ (if any)  using \SMerge{},
and finally inserted into buffer $U_{0}$. After the merge, the immediate
insertion buffer, the priority queue $P_{I}$ and all its associated buffers  are
emptied. If the merge does not cause $U_{0}$ to overflow, the priority queue
$P_{U}$  is rebuilt from the new values in $U_0^S$  by following the previously
described procedure.
On the contrary, if $U_{0}$ overflows breaking the size invariant (I4), the
\Push{} primitive is invoked on $U_{0}$, $P_{U}$ is deallocated (since \Push{}
removes all keys in $U_0$) and $P_{D}$ is rebuilt following a procedure similar
to the one for $P_{U}$.

\paragraph{\Deletemin{}} To determine and remove the minimum key in the priority
queue it is necessary to evaluate the minimum key among the at most $\log n +
\delta / S$ keys in the immediate insertion buffer $I_{0}$ and the minimum
values in $P_{I}$, $P_{D}$ and $P_{U}$, which can be evaluated using \Peekmin{}.
Finally, the minimum key $v$ among these four values is selected, removed from
the appropriate buffer as described below, and hence returned. The removal of
$v$ is performed as follows.

\begin{itemize}
 \item \emph{$v$ is in $I_{0}$.} Value $v$ is removed
from $I_0$ and the remaining keys in $I_0$ are shifted in order to 
ensure that keys are consecutively stored.

\item \emph{$v$ is in  $P_{I}$.}  A \Deletemin{} is performed on $P_I$ for
removing the node with key $v$. Let $I'$ be the buffer pointed by this node.
Then key $v$ is removed from $I'$ and the remaining keys in $I'$ are shifted in
order to ensure that keys are consecutively stored. We note that the value $v$
may not be  anymore available in $I'$ since it has been corrupted by the
adversary: however, since each node contains the position of $v$ in $I'$, the
faithful value can be restored. Let $c_{I'}$ be the new size of $I'$.  If
$c_{I'} \geq (\log n + \delta / S)/2$, a new node pointing to $I'$ is inserted
in $P_I$ using as priority key the new minimum value in $I'$. If $c_{I'} < (\log
n + \delta / S)/2$ and $I_{0}$ is not empty, up to $(\log n + \delta / S)/2$
keys are removed from the immediate insertion buffer and inserted in buffer
$I'$; then, a new node is inserted in $P_I$ pointing to $I'$ and with priority
key set to the new minimum value in $I'$. Finally, if $c_{I'} < (\log n + \delta
/ S)/2$ and $I_{0}$ is empty, all values in $I'$ are transferred in the
immediate insertion buffer $I_0$ and $I'$ is deallocated.

\item \emph{$v$ is in  $P_{U}$.} A \Deletemin{} is performed on $P_{U}$ for
removing the node with key $v$. Let $U'$ denote the sub-buffer  pointed by the
removed node. The minimum key $v$ is removed from $U'$ and its spot is filled
with the value pointed by  $p_{U}$, which is then increased to point to the
subsequent value in $U^P_0$ (if any). If no key can moved to $U'$ (i.e., there
are no keys in $U_0^P$), the empty spot is removed by compacting $U'$ in
order to ensure that keys are consecutively stored and no further operations are
performed. The new minimum value in $U'$ is then evaluated and inserted in
$P_{U}$ with the associated pointer to $U'$ (no operation is done if $U'$ is
empty). 

\item \emph{$v$ is in  $P_{D}$.} 
Operations similar to the previous case are performed if the minimum key
is extracted from $P_{D}$. In this case, \Deletemin{} may cause $D_{0}$ to
underflow breaking the
size invariant I3: if that happens, the \Pull{} primitive is invoked on $D_{0}$
and $P_{D}$ is rebuilt following a procedure analogous to the one previously
detailed for $P_U$. 
\end{itemize}

We observe that the use of the auxiliary structures $P_U$ and $P_D$ in
\Deletemin{} may cause a discrepancy with
respect to the order invariants (I1) and (I2) for buffers $U_{0}$ and $D_{0}$.
We can however justify the waiver from (I1) by pointing out that this
structure still ensures that the faithful keys in $U^{S}_{0}$ are smaller than
or equal to those in $U^{P}_{0}$. In particular the concatenation of a
resiliently sorted $U^{S}_{0}$ with $U^{P}_{0}$ is faithfully ordered
(similarly in $D_0$).
Additionally, we can justify the waiver from (I2) by
observing that the faithful keys in $D_{0}$
are still smaller than or equal to those in $D_{1}$ and $U_{1}$. Furthermore,
the invariants can be easily restored before any invocation of \Push{} and
\Pull{} by resiliently sorting $U^{S}_{0}$ (resp., $D^{S}_{0}$) and
linking it with $U^{P}_{0}$ (resp., $D^{P}_{0}$).
Therefore, since  $D_{0}$ and $U_{0}$ still behave
consistently with the invariants for what pertains the relations with other
buffers and the possibility of accessing the faithful keys maintained by them in
the correct order,  we can assume with a slight (but harmless) ``abuse of
notation'' that the invariants are verified for $U_{0}$ and $D_{0}$ as well.

\subsubsection{\Push{} and \Pull{} primitives}\label{sec:pushpull}
\Push{} and \Pull{} are the two fundamental primitives used to structure and
maintain the resilient priority queue. Their execution is triggered whenever one
of the buffer violates a size invariant in order to restore it without affecting
the order invariants. The primitives operate by redistributing keys among
buffers by making use of \SMerge{}. The main idea is to move keys in the buffers
in order to have the smaller ones kept in the layers close to the insertion
buffer so they can be quickly retrieved by \Deletemin{} operations,  while
moving the larger keys to the higher order layers. Our implementation of \Push{}
and \Pull{} corresponds to the one in~\cite{JMM07} with the difference that the
\SMerge{} algorithm proposed in the previous section is used rather than the
merge algorithm in~\cite{FGI09}. It is important to stress how this variation,
while allowing a reduction of the running time of \Push{} and \Pull{}, does not
affect the correctness nor the functioning of the primitives since the merge
algorithm is used with a \emph{black-box} approach. We remark that, due to the
additional structure introduced by using the auxiliary priority queues $P_U$ and
$P_D$, every time a primitive involving either $U_0$ or $D_0$  is invoked it is
necessary to restore them to be faithfully ordered buffers. This can be easily
achieved by concatenating the resiliently sorted $U^P_0$ (resp., $D^S_0$), with
$U^S_0$ (resp., $D^S_0$). We can exploit $P_U$ (resp., $P_D$) to resiliently
sort  $U^P_0$ (resp., $D^S_0$) by successively extracting the minimum values in
the priority queue. For the sake of completeness, we describe now the \Push\ and
\Pull\ primitives and we refer to~\cite{JMM07} for further details.

\paragraph{\Push{}} The \Push{} primitive is invoked whenever the size of an up
buffer $U_{i}$ grows over the threshold value $s_{i}/2$, therefore breaking the
size invariant (I4). 
The execution of \Push{}($U_{i}$) works as follows.
If $L_i$ is the last layer, then a new empty layer $L_{i+1}$ is
created. Buffers $U_{i}$, $D_{i}$ and
$U_{i+1}$ are merged into a sequence $M$ using the \SMerge{} algorithm. Then
the first $\vert D_{i} \vert - \delta$ keys of $M$ are placed in a new buffer
$D'_{i}$, the remaining $\vert U_{i+1} \vert + \vert U_{i} \vert + \delta$
keys are placed in a new buffer $U'_{i+1}$, and  an empty $U'_{i}$
buffer is created. Finally, the newly created buffers $U'_{i}$, $D'_{i}$ and
$U'_{i+1}$ are  used to respectively replace the old buffers $U_{i}$,
$D_{i}$ and $U_{i+1}$, which are then deallocated. If $L_{i}$ is the last layer,
$U'_{i+1}$ replaces $D_{i+1}$ instead of $U_{i+1}$. If the new buffer $U'_{i+1}$
contains too many keys, breaking  the size invariant (I4), the \Push{} primitive
is invoked on $U'_{i+1}$. Furthermore, since $D'_{i}$ is smaller than $D_{i}$,
it could violate the size invariant (I3). This violation is handled at the end
of the sequence of \Push{} invocations on up buffers of layers
$L_{i},L_{i+1},\ldots,L_{j}$, $0 \leq i < j < k$ (we suppose the $i$ and $j$
indexes to be stored in safe memory).  After all the $j-i+1$ invocations, the
affected down buffers are analyzed by simply following the pointers among
buffers starting from $U_{i}$, and by invoking the \Pull{} primitive (see below)
on the down buffer not satisfying the invariant (I3). 

\paragraph{\Pull{}} The \Pull{} primitive is invoked whenever the size of a down
buffer $D_{i}$ goes below the threshold value $s_{i}/2$, therefore breaking the
size invariant (I3). Since this invariant does not hold for the last layer, we
must have that $L_i$ is not the last layer. During the execution of
\Pull{}($D_{i}$) buffers $D_{i}$, $U_{i+1}$, and $D_{i+1}$ are merged into a
sequence $M$ using the \SMerge{} algorithm. The first $s_{i}$ keys of $M$ are
placed in a new buffer $D'_{i}$, the following $\vert D_{i+1} \vert -
\left(s_{i} - \vert D_{i}\vert \right)-\delta$ keys are written to $D'_{i+1}$,
while the remaining keys in $M$ are placed in a new buffer $U'_{i+1}$. The newly
created buffers $D'_{i}$, $D'_{i+1}$ and $U'_{i+1}$ are then used to
respectively replace the old buffers $D_{i}$, $D_{i+1}$ and $U_{i+1}$, which are
then deallocated. If the down and up buffers in layer $L_{i+1}$ are empty after
this operation, then layer  $L_{i+1}$ is removed (this can happen only if
$L_{i+1}$ is the last layer). Resulting from this operation, $D'_{i+1}$ may
break the size invariant (I3), if this is the case \Pull{} is invoked on
$D'_{i+1}$. Additionally, after the merge, $U′_{i+1}$ may break the size
invariant (I4). 
This violation is handled at the end of the
sequence of \Pull{} invocations on down buffers of layers
$L_{i},L_{i+1},\ldots,L_{j}$, $0
\leq i < j < k$ (we suppose the $i$ and $j$ indexes to be stored in safe
memory).  After all the $j-i+1$ invocations, all the affected up buffers are
analyzed by
simply following the pointers among buffers starting from $D_{i}$, and by
invoking the \Push{} primitive wherever invariant (I4) is not
satisfied.

\subsection{Correctness and complexity analysis}\label{sec:corr}
In order to prove the correctness of the proposed resilient priority queue we
show that \Deletemin{} returns the minimum faithful key in the priority queue or
an even smaller corrupted value. As a first step, it is necessary to ensure that
the invocation of one of the primitives \Push{} or \Pull{}, triggered by an up
or down buffer violating a size invariant I3 or I4, does not cause the order
invariants to be broken. The \Push{} and \Pull{} primitives used in our priority
queue coincide with the ones presented for the maintenance of the resilient
priority queue in~\cite{JMM07}: despite the fact that in our
implementation the threshold $s_{i}$ is changed to $2^{i+1}\left( S\log^{2}+
\delta \left( \log S + \delta / S \right) \right)$, the proofs provided
in~\cite{JMM07} (Lemmas 1 and 3) concerning the correctness of \Push{} and
\Pull{} still apply in our case. We report here the statements of the cited
lemmas: 

\begin{lemma}[{\cite[Lemma~1]{JMM07}}]\label{lem:1}
The \Pull{} and \Push{} primitives preserve the order invariants.
\end{lemma}

\begin{lemma}[{\cite[Lemma~3]{JMM07}}]\label{lem:2}
If a size invariant is broken for a buffer in $L_0$, invoking \Pull{} or \Push{}
on that buffer restores the invariants. Furthermore, during this operation
\Pull{} and \Push{} are invoked on the same buffer at most once. No other
invariants are broken before or after this operation.
\end{lemma}

For the complete proofs of these lemmas we refer the reader to the original work
in~\cite{JMM07}. It is important to remark that both proofs are independent of
the value used as size threshold and hence these proofs hold for our
implementation as well.
We can therefore conclude that when a size invariant is broken for a buffer in
$L_{i}$ the consequent invocation of \Push{} or \Pull{} does indeed restore the
size invariant while preserving the order invariants which are thus maintained
at all times.

Concerning the computational cost of the primitives, an analysis carried out
using the potential function method~\cite[Section 17.3]{Cormen09} allows to
conclude that the amortized time
needed for the execution of both \Push{} and \Pull{} is negligible. A proof of
this fact can be obtained by plugging the complexity of the \SMerge{} algorithm 
and the threshold value $s_i$ defined in our implementation in the proof
proposed in~\cite[Lemma 5]{JMM07}.

\begin{lemma}\label{lem:pushspull}
The amortized cost of the \Push{} and \Pull{} primitives is negligible. 
\end{lemma}
\begin{proof}
We now upper bound the amortized cost of a call to the \Push\ function  on the
up buffer $U_i$ and we ignore at the moment the subsequent chain of calls to
\Push\ and \Pull  (a similar argument applies to \Pull).  The cost is 
computed by exploiting the
following \emph{potential function} defined in~\cite{JMM07}:
\begin{equation*}
\Phi = \sum_{i=1}^k \left( c_{1}  \vert U_{i} \vert\left( \log n - i\right) +
ic_{2}\vert D_{i} \vert \right).
\end{equation*}
When a \Push{} operation on $U_{i}$ is performed, first the $U_{i}$, $D_{i}$ and
$U_{i+1}$ buffers are merged and then the sorted values are distributed into new
buffers such that $\vert U'_{i} \vert = 0$, $\vert D'_{i} \vert = \vert D_{i}
\vert - \delta $ and $\vert U'_{i+1} \vert= \vert U_{i+1} \vert + \vert U_{i}
\vert + \delta
$. This leads to the following change in potential $\Delta \Phi$:
\begin{equation*}
\begin{split}
\Delta \Phi &= - c_{1} \vert U_{i} \vert \left( \log n -i \right) - i
c_{2}\delta  + c_{1}\left( \vert U_{i} \vert + \delta \right) \left( \log n -
\left( i+1
\right)\right)\\
&=  - c_{1} \vert U_{i} \vert +\delta \left(  - i c_{2}  + c_{1} \log n -
i c_{1} - c_{1} \right).
\end{split}
\end{equation*}
\Push{} is invoked when (I4) is not valid for $U_{i}$ and therefore $\vert U_{i}
\vert > s^{i}/2 = 2^{i}\left(S \log^{2} n + \delta \left( \log S + \delta / S
\right) \right)$. Then, standard computations show that, for some constant $c' >
0$ independent of $c_1$, we have
\begin{equation*}
\Delta \Phi \leq - c_{1} \vert U_{i} \vert + c_{1}\delta \log n \leq
-c_{1}c'\vert U_{i} \vert.
\end{equation*}
The time required for the execution of \Push{}, including the time needed to
retrieve the reliably stored pointers of the up and down buffers, is dominated
by the computational cost of merging $U_{i}$, $D_{i}$ and $U_{i+1}$ which, using
the \SMerge{} algorithm, is upper bounded by $T_m = \BO{ \vert U_{i} \vert +
\vert D_{i} \vert + \vert U_{i+1} \vert + \alpha \left( \log S + \delta / S
\right)}$. By the potential method~\cite[Section 17.3]{Cormen09}, the amortized
cost of \Push{}  follows by adding the merging time $T_m$ to the potential
variation $\Delta \Phi$. Since $\vert U_{i} \vert \in \BT{2^i\left(S \log^{2} n
+ \delta \left( \log S + \delta / S \right)\right)} $, we have
$T_m=\BT{|U_i|}=c_m |U_i|$, where $c_{m}$ is a suitable constant that depends on
\SMerge. The amortized cost of \Push{} is $(c_m-c_1c')|U_i|$ and it can be
ignored by conveniently tweaking $c_{1}$ according to the values of $c_m$ and
$c'$ so that $(c_m-c_1c')|U_i|<0$.

In the particular case for which an
invocation of \Push{} involves the buffers $U_0$ and $D_0$, we have that prior
to the standard operations, it is necessary to restore the buffers to
their faithfully sorted version by resiliently sorting $U^S_0$ (resp.,
$D^S_0$) and linking it with $U^P_0$ (resp., $D^P_0$). The time required
to accomplish these operation is dominated by the time necessary to faithfully
sort $U^S_0$ and $D^S_0$ according to the previously described technique, which
is $\BO{\delta\left(\log S + \delta/S\right)}$. This implies that the time
required
for restructuring $U_0$ and $D_0$ is sill dominated by the time required
by \Push{} and is hence negligible.

Since each \Push\ and \Pull\ function is invoked on the
same buffer at most once (Lemma~\ref{lem:2}) and the amortized cost is
negative, we have that the chain of \Push\ and \Pull\ operations that can
start after the initial call is  negligible as well.  The lemma follows.

\end{proof}

The following theorem evaluates the amortized cost of \Insert{} and
\Deletemin{} in our resilient implementation of the priority queue. 

\begin{theorem}
In the proposed resilient priority queue implementation, the \Deletemin{}
operation returns the minimum faithful key in the priority queue or an even
smaller corrupted one and deletes it. Both \Deletemin{} and \Insert{} operations
require $\BO{\log n + \delta / S}$ amortized time. The priority queue uses
$\BT{S}$   safe memory words and $\BT{n}$ faulty memory words.
\end{theorem}
\begin{proof}
We first observe that the size and order invariants can be
considered maintained at all times thanks to the maintenance \Push\ and \Pull\
tasks (see Lemmas~\ref{lem:1} and~\ref{lem:2}), with the aforementioned
exception on the first $\delta+1$ keys in the up and down buffers in $L_0$. 
Moreover, by Lemma~\ref{lem:pushspull},  the cost of \Push{} and \Pull{} can be
ignored in our argument.

We now focus on the correctness and complexity of
\Deletemin{}. Let $v_1$, $v_2$, $v_3$ and $v_4$ be the minimum values in $I_0$,
$P_I$, $P_U$ and $P_D$, respectively. \Deletemin{} evaluates these four values
by scanning all the values in $I_0$ and by performing a \Peekmin{} operation for
$P_I$, $P_U$ and $P_D$, respectively. By construction, each value in $P_{I}$ is
selected as the minimum among the keys stored in the associated buffers: since
$P_{I}$ is maintained in the safe memory, $v_2$ is smaller than any faithful
value in the associated buffers. Similarly, $v_3$ is smaller than the faithful
$\delta+1$ entries in $U_0^S$, and thus of the remaining faithful entries in
$U^P_0$ and of all entries in the up and down buffers for invariant (I1).
Similarly, we also
have that $v_4$ is smaller than all faithful keys in $D_0$. We can then conclude
that $\min\{v_1, v_2, v_3, v_4\}$ is either the minimum faithful key in the
priority queue or an even smaller corrupted value.  The time for determining the
minimum key and removing it is $\BO{\log n + \delta / S}$.

We now discuss  the correctness and complexity of \Insert{}. The correctness of
the  insertion is evident since the input key is inserted in some support buffer
and can be only removed by \Deletemin. Inserting a key in the immediate
insertion buffer requires constant time. If $I_{0}$ is full and a new node of
the priority queue $P_I$ needs to be created, a total $\BO{\log n + \delta / S
}$ time is required in order to find the minimum among the keys in $I_{0}$ and
to insert the new node in $P_{I}$. When $P_I$ itself is full (i.e., contains $S$
nodes), we have that $\BO{S \log^{2} n + \delta\left( \log n + \delta / S
\right)}$ time is  required to faithfully sort all keys in $I_0$ and in the
buffers managed through $P_I$ and $P_U$, to faithfully merge them with
$U^S_{0}$, and to rebuild $P_{U}$ and $P_S$. However, it will be necessary to
perform these operations at most once every $\BT{S\log n +\delta}$ key
insertions and therefore its amortized cost is $\BO{\log n + \delta / S}$. 

We recall that the algorithm invokes a global rebuilding every time the number
of keys changes by a $\BT{n}$ factor. Since the cost of the rebuilding is
dominated by the cost of the \SSort{} algorithm, which is $\BO{n\log n+\delta
(\delta/S +  \log S)}$, the amortized cost is $\BO{\log n + \delta/S}$.

By opportunely doubling or halving the space reserved for the immediate buffer
$I_0$, the space required for $I_0$ is always at most twice the number of keys
actually in the buffer. Additionally, the space required  for the buffers 
maintained by $P_I$ is at most double than the number of keys actually in the
buffer itself. The space required for each layer $L_{0},\ldots,L_{k-1}$ with $k
\in \BO{\log n}$, including the reliably written structural information, is
proportional to the number of stored keys, and therefore $\BT{n}$ faulty
memory words are used to store all the layers. Finally, $\BT{S}$ safe memory
words are required to maintain the priority queues $P_{I}$, $P_{D}$ and $P_{U}$
and for the correct execution of \SMerge{} and \SSort. The theorem follows.
\end{proof}

\section{Conclusion}
In this paper we have shown that, for the resilient sorting problem and the
priority queue data structure, the presence of a safe memory of size
$S$ can be exploited in order to reduce the computational overhead due to the
presence of corrupted values by a factor $\BT{S}$. As future research, it
would be interesting to investigate which other problems can benefit of a non
constant safe memory and propose tradeoffs highlighting the achievable
performance with respect to the size of the available safe memory. We observe
that not all problems can in fact exploit an $S$-size safe
memory: indeed the the $\BOM{\log n +\delta}$ lower bound for searching derived
in \cite{FI08}  applies even if a safe memory of size $S\leq \epsilon n$,
for a suitable constant $\epsilon\in (0,1)$, is available. Finally, we remark
that the analysis of tradeoffs between the safe memory size and
the performance achievable by resilient algorithms may provide useful insights
for designing hybrid systems mounting both cheap faulty memory and expensive ECC
memory, as recently studied in~\cite{LI13}.

\section*{Acknowledgements}
The authors would like to thank G. Brodal, I. Finocchi and an anonymous
reviewer for useful
comments. This work was supported, in part, by University of Padova under
projects STPD08JA32 and CPDA121378, and by MIUR of Italy under project AMANDA.

\end{document}